\documentclass[a4paper,USenglish,cleveref,autoref,thm-restate]{lipics-v2021}

\usepackage{hyperref}
\usepackage{comment}

\newcommand{\tightcite}[1]{\kern-0.15em\cite{#1}}

\usepackage[backgroundcolor=white,textcolor=orange]{todonotes}
\setuptodonotes{inline}

\usepackage{cite}

\usepackage{cleveref}

\nolinenumbers
\hideLIPIcs

\graphicspath{{./figures/}}

\def\computationproblem#1#2#3{
  \begin{center}
  \begin{tabular}{rp{0.8\textwidth}}
  {\sc Problem:\enspace}&#1\\
  {\sc Input:\enspace}&#2\\
  {\sc Question:\enspace}&#3\\
  \end{tabular}
  \end{center}
}

\def\na#1{^{\uparrow #1}}

\newtheorem*{conjecture2}{Conjecture}

\title{Computational complexity of covering regular trees} 

\titlerunning{Computational complexity of covering regular trees}

\author{Jan Bok}{Department of Algebra, Faculty of Mathematics and Physics, Charles University, Prague, Czech Republic}{jan.bok@matfyz.cuni.cz}{https://orcid.org/0000-0002-7973-1361}{Supported by the European Union (ERC, POCOCOP, 101071674). Views and opinions expressed are however those of the author(s) only and do
not necessarily reflect those of the European Union or the European Research Council Executive Agency. Neither the European
Union nor the granting authority can be held responsible for them.}
\author{Jiří Fiala}{Department of Applied Mathematics, Faculty of Mathematics and Physics, Charles University, Prague, Czech Republic}{fiala@kam.mff.cuni.cz}{https://orcid.org/0000-0002-8108-567X}{Supported by research grant GAČR 20-15576S of the Czech Science Foundation.} 
\author{Nikola Jedličková}{Department of Applied Mathematics, Faculty of Mathematics and Physics, Charles University, Prague, Czech Republic \and Department of Algebra, Faculty of Mathematics and Physics, Charles University, Prague, Czech Republic}
{jedlickova@kam.mff.cuni.cz}{https://orcid.org/0000-0001-9518-6386}{Supported by the research grants PRIMUS/24/SCI/008, UNCE/24/SCI/022 of Charles University, and SVV–2025–260822.}
\author{Jan Kratochvíl}{Department of Applied Mathematics, Faculty of Mathematics and Physics, Charles University, Prague, Czech Republic}{honza@kam.mff.cuni.cz}{https://orcid.org/0000-0002-2620-6133}{Supported by research grant GAČR 20-15576S of the Czech Science Foundation.}

\authorrunning{J. Bok, J. Fiala, N. Jedličková, J. Kratochvíl}

\Copyright{Jan Bok and Jiří Fiala and Nikola Jedličková and Jan Kratochvíl} 

\ccsdesc[500]{Mathematics of computing~Graph theory}

\keywords{graph cover, covering projection, semi-edges, multigraphs, complexity, constrained homomorphisms, trees}

\EventEditors{Pawe\l{} Gawrychowski, Filip Mazowiecki, and Micha\l{} Skrzypczak}
\EventNoEds{3}
\EventLongTitle{50th International Symposium on Mathematical Foundations of Computer Science (MFCS 2025)}
\EventShortTitle{MFCS 2025}
\EventAcronym{MFCS}
\EventYear{2025}
\EventDate{August 25--29, 2025}
\EventLocation{Warsaw, Poland}
\EventLogo{}
\SeriesVolume{345}
\ArticleNo{47}

\begin{document}

\maketitle

\begin{abstract}
A graph covering projection, also referred to as a locally bijective homomorphism, is a mapping between the vertices and edges of two graphs that preserves incidences and is a local bijection. This concept originates in topological graph theory but has also found applications in combinatorics and theoretical computer science. In this paper we consider undirected graphs in the most general setting -- graphs may contain multiple edges, loops, and semi-edges. This is in line with recent trends in topological graph theory and mathematical physics. 

We advance the study of the computational complexity of the {\sc $H$-Cover} problem, which asks whether an input graph allows a covering projection onto a parameter graph $H$. The quest for a complete characterization started in 1990's. Several results for simple graphs or graphs without semi-edges have been known, the role of semi-edges in the complexity setting has started to be investigated only recently. One of the most general known NP-hardness results states that {\sc $H$}-Cover is NP-complete for every simple connected regular graph of valency greater than two. We complement this result by considering regular graphs $H$ arising from connected acyclic graphs by adding semi-edges.
Namely, we prove that any graph obtained by adding semi-edges to the vertices of a tree  making it a $d$-regular graph with $d \geq 3$, defines an NP-complete graph covering problem. In line with the so called Strong Dichotomy Conjecture, we prove that the NP-hardness holds even for simple graphs on input.
\end{abstract}

\newpage
\section{Introduction}\label{sec:Intro}

The study of graph homomorphisms has a long-standing tradition in combinatorics and theoretical computer science. During the past few decades, numerous variations of the problem have been considered.
Informally, graph homomorphism is an adjacency-preserving mapping between two graphs. 
A significant portion of research in this area focuses on so-called \emph{constrained homomorphisms}, where additional local or global conditions are imposed on the
mapping, such as bijectivity, surjectivity, or injectivity.
Among those, the most widely studied are locally bijective homomorphisms, also called \emph{graph coverings} (or \emph{graph lifts}). The notion comes from topology, where a covering (also referred to as a covering projection) is a special type of local homeomorphism between topological spaces, and the term of graph covering is basically a discretization of this important and deeply studied concept.

The notion of graph covering quickly proved useful in discrete mathematics, namely as a tool to construct highly symmetric graphs~\cite{n:Djokovic74,n:Gardiner74,n:Biggs74,n:Biggs81,n:Biggs82,n:Biggs84} or for embedding complete graphs in surfaces of higher genus~\cite{k:Ringel74}. In theoretical computer science, the notion surprisingly found applications in the theory of so-called local computations~\cite{n:Angluin80,n:Chalopin05,n:ChMZ06,n:ChalopinP11,n:CM94,n:LMZ93}. In~\cite{packing_bipartite}, the authors study a close relation of packing bipartite graphs to a special variant of graph coverings called \emph{pseudo-coverings}.

Despite the efforts and attention that graph covers have received in the computer science community, their computational complexity remains far from being fully understood.
Bodlaender~\cite{n:Bodlaender89} proved that deciding if one graph covers another is an NP-complete problem if both graphs are part of the input. The question is obviously at least as difficult as the famous graph isomorphism problem: consider two given graphs on the same number of vertices. Abello, Fellows, and Stillwell~\cite{n:AFS91} later considered the variant of having the target graph, say $H$, fixed and studied the following parameterized problem.

\computationproblem{{\sc $H$-Cover}}{A graph $G$.}{Does $G$ cover $H$?}

They were the first to raise the question of a complete characterization of the computational complexity of this problem. This initiated a substantial body of research, primarily focusing on graphs without semi-edges (while already in the seminal paper \cite{n:AFS91} graphs with multiple edges and loops have been considered). This includes proving the polynomial-time solvability of {\sc $H$-Cover} for connected simple undirected graphs $H$ that have at most two vertices in every equivalence class of their degree partitions~\cite{n:KPT94} (a graph is simple if it contains no loops, no semi-edges and no multiple edges). Furthermore, the complexity of covering $2$-vertex multigraphs was fully characterized in~\cite{n:KPT97a}, the characterization for 3-vertex undirected multigraphs can be found in~\cite{n:KTT16}. The most general NP-hardness result known so far is the hardness of covering simple regular graphs of valency at least three~\cite{n:Fiala00b,n:KPT97}. More recently, B\'{\i}lka et al.~\cite{n:BilkaJKTV11} proved that covering several concrete small graphs (including the complete graphs $K_4, K_5$ and $K_6$) remains NP-hard for planar inputs. This shows that planarity does not help in graph covering problems in general, yet the conjecture that the \textsc{$H$-Cover} problem restricted to planar inputs is at least as difficult as for general inputs, provided $H$ itself has a finite planar cover, remains still open. Planar graphs have also been considered by Fiala et al.~\cite{n:FialaKKN14} who showed that for planar graphs $H$, the so called \textsc{$H$-RegularCover} is in FPT when parameterized by $H$. Let us point out that in all the above results it was assumed that $H$ has no multiple edges, no loops and no semi-edges.

Another connection to algorithmic theory comes through the notions of the {\em degree partition} and the {\em degree refinement matrix} of a graph. These notions were introduced by Corneil~\cite{Corneil68,n:CorneilG70} in hope of solving the graph isomorphism problem efficiently. This motivated Angluin and Gardiner~\cite{n:AG81} to prove that any two finite regular graphs of the same valency have a finite common cover, and conjecture the same for every two finite graphs with the same degree refinement matrix, which was later proved by Leighton~\cite{n:Leighton82}.

The stress on finiteness of the common cover is natural. For every matrix, there exists a universal cover, an infinite tree, that covers all graphs with this degree refinement matrix.  Trees are planar graphs, and this inspired an at first sight innocent question of which graphs allow a finite planar cover.  Negami observed that projective planar graphs do (in fact, their double planar covers characterize their projective embedding), and conjectured that these two classes actually coincide~\cite{n:Negami88}. Despite a serious effort of numerous authors, the problem is still open, although the scope for possible failure of Negami's conjecture has been significantly reduced \cite{n:Archdeacon02,n:Hlineny98,n:HT04}.

As mentioned earlier, the requirement for local bijectivity can be relaxed. In these relaxations, homomorphisms can be either locally injective or locally surjective and not necessarily locally bijective.  The computational complexity of locally surjective homomorphisms has been classified completely, with respect to the fixed target graph~\cite{n:FP05}. Though the complete classification of the complexity of locally injective homomorphisms is still out of sight, it has been proved for its list variant~\cite{n:FK06}.  The problem is also interesting for its applied motivation --- the locally surjective version has played an important role in the social sciences~\cite{n:FP05}, particularly in the \emph{Role Assignment Problem}. A prominent special case of the locally injective homomorphism problem is the well-studied $L(2,1)$-labeling problem~\cite{n:GriggsY92}. Further generalizations include the notion of $H(p,q)$-coloring, a homomorphism into a fixed target graph $H$ with additional rules on the neighborhoods of the vertices~\cite{n:FHKT03,n:KT00}. It is worth noting that for every fixed graph $H$, the existence of a locally injective homomorphism to $H$ is provably at least as hard as the $H$-cover problem itself. To find more about locally injective homomorphisms, see e.g.~\cite{n:ChaplickFHPT15,n:FK08,n:MacGillivrayS10}.
We also refer the reader to the survey concerning various aspects of locally constrained homomorphisms~\cite{n:FK08}.

So far, we have not been speaking about \emph{semi-edges}, which are informally really just edges with exactly one endpoint (similar to loops, which have a single unique endpoint but are incident twice to that endpoint). Semi-edges appear naturally in group-theoretical and topological constructions where covers are mostly used and employed. Graphs with semi-edges arise as quotients of standard graphs by groups of automorphisms which are semiregular on vertices and darts (arcs) but may fix edges. From the graph-theoretical point of view, semi-edges allow an elegant and unifying description of several crucial graph-theoretical notions. For instance, it was observed before that a simple $k$-regular graph is $k$-edge-colorable if and only if it covers the one-vertex graph with $k$ semi-edges. Similarly, a simple cubic graph contains a perfect matching if and only if it covers the one-vertex graph with one loop and one semi-edge.  Last but not least, in modern topological graph theory graphs with semi-edges are widely used and considered since these occur naturally in algebraic graph reductions. Let us name just a few other significant examples of usage of semi-edges. Malni{\v{c}} et al.~\cite{n:MalnivcMP04} considered
semi-edges during their study of abelian covers to allow for a broader range of applications. 
Furthermore, the concept of graphs with semi-edges was introduced independently and
naturally in mathematical physics~\cite{getzler1998modular}.
It is also natural to consider semi-edges in the mentioned framework of local computations (we refer to the introductory section of~\cite{n:BFHJK21-MFCS} for more details). Finally, a theorem of Leighton~\cite{n:Leighton82} on finite common covers has been recently generalized to the semi-edge setting in~\cite{arxiv1908.00830,woodhouse_2021}. To highlight a few other contributions, the reader is invited
to consult~\cite{n:MalnicNS00,n:NedelaS96}, the surveys~\cite{kwak2007graphs,nedela_mednykh}, and
finally for more recent results, the series of papers~\cite{n:FialaKKN14,n:FialaKKN18,arxiv1609.03013} and the introductions therein. 

Surprisingly, the study of the computational complexity of covering graphs with semi-edges is quite young and recent and it was initiated by Bok et al.\ in~\cite{n:BFHJK21-MFCS} in 2021 where the complete complexity dichotomy was established for one- and two-vertex target graphs.
Continuing in this direction, Bok et al.~\cite{n:BFJKS21-FCT} studied the computational complexity of covering disconnected graphs and showed that, under an adequate yet natural definition of covers of disconnected graphs, the existence of a covering projection onto a disconnected target graph is polynomial-time decidable if an only if it is the case for every connected component of the target graph.
In \cite{n:BFJKR24}, the authors consider the list version of the problem, called \textsc{List-$H$-Cover}, where the vertices and edges of the input graph come with lists of admissible targets. It is proved that the \textsc{List-$H$-Cover} problem is NP-complete for every regular graph $H$ of valency greater than 2 which contains at least one semi-simple vertex (i.e., a vertex which is incident with no loops, with no multiple edges and with at most one semi-edge).
Using this result they show the NP-co/polytime dichotomy for the computational complexity of \textsc{ List-$H$-Cover} for cubic graphs $H$.  
Most recently, Bok et al.~\cite{n:BFJKS23-WG} presented a complete characterization of the computational complexity of covering colored mixed (multi)graphs with semi-edges whose every class of degree partition contains at most two vertices. That provides a common generalization of the observation for simple graphs from~\cite{n:KPT97a}, and for the case of regular 2-vertex graphs from~\cite{n:BFHJK21-MFCS}. 

An interesting fact is that all the known NP-hard instances of {\sc $H$-Cover} remain NP-hard for simple input graphs. This led the authors of~\cite{n:BFJKR22-IWOCA} to formulate the following conjecture.

\begin{conjecture2}[Strong Dichotomy Conjecture for Graph Covers~\cite{n:BFJKR22-IWOCA}]
For every graph $H$, the problem {\sc $H$-Cover} is either polynomial-time solvable for arbitrary input graphs, or it is NP-complete for simple graphs as input.
\end{conjecture2}

In the following subsection, we present our main results, which align well with the Strong Dichotomy Conjecture.

\subsection*{Our results and motivation}\label{sub:results}

Suppose that $T$ is a (simple) tree of maximum degree $\Delta$. For $d\ge \Delta$, denote by $T\na{d}$ the $d$-regular graph obtained from $T$ by adding semi-edges (which means that a vertex of degree $k$ in $T$ gets $d-k$ semi-edges added to it).  With a slight abuse of notation we call $T\na{d}$ a \emph{$d$-regular tree}. Our goal is to prove the following theorem.

\begin{theorem}\label{thm:main}
For every simple tree $T$ of maximum degree $\Delta$ and every $d\ge \max\{\Delta,3\}$, the problem {\sc $T\na{d}$-Cover} is NP-complete, even for simple input graphs.
\end{theorem}

Our motivation is twofold. First, the complexity of {\sc $H$-Cover} for regular simple graphs of valency at least three is known for a long time \cite{n:Fiala00b,n:KPT97}. The study of the complexity of {\sc $H$-Cover} for regular graphs $H$ with semi-edges was initiated only recently~\cite{n:BFHJK21-MFCS,n:BFJKR24}. A general result was presented, but it relies on the use of lists.
Second, simple trees are easy instances for {\sc $H$-Cover} as they correspond to isomorphisms~\cite{k:Reidemeister32}. However, we show that when trees become regularized by adding semi-edges, they yield NP-complete {\sc $H$-Cover} problems. This is a striking difference from the polynomiality of checking whether a graph covers (which means it is isomorphic to) a given tree.

The rest of the paper is organized as follows. In Section~\ref{sec:Prelim}, we review the formal definitions of undirected graphs (with multiple edges, loops and semi-edges allowed) and of covering projections between such graphs. For the convenience of the reader, we first review the definition of covering projections between simple graphs. Then we list basic graph theory notions and their meaning in the general setting, namely when semi-edges are present in the graph under consideration. We also gather several auxiliary results from other sources that will be useful in the rest of the paper.

The proof of Theorem~\ref{thm:main} is presented in the following three sections. The most general case of $d\ge 4$ is treated en bloc in Section~\ref{sec:4-regular}. For $d=3$, we distinguish two cases. If the underlying simple tree $T$ has maximum degree 3, i.e., if $d=\Delta(T)=3$, we show that an approach used in \cite{n:BFJKR22-IWOCA} for showing NP-hardness of the list covering version in the case that the parameter graph has a semi-simple vertex can be utilized in our setting. This is shown in Section~\ref{sec:deg3}. The case of $d=3$ and $\Delta(T)=2$, i.e., when the underlying simple tree is a path, is singled out in Section~\ref{sec:Paths}. The last section then contains concluding remarks.

\section{Preliminaries}\label{sec:Prelim}

For integers $i$, $j$, the notation $[i,j]$ stands for the set $\{i,i+1,\dots,j\}$. It is an empty set, when $i>j$. We often abbreviate $[1,j]$ by $[j]$.

In this paper we only consider undirected graphs without further stressing this. A \emph{simple graph} is a pair $G=(V,E)$, where $V$ is a set of \emph{vertices} and $E\subseteq \binom{V}{2}$ is a set of \emph{edges}. Edges are viewed as two-element subsets of the vertex set, but for the sake of brevity, we will often write $uv$ as a short-hand notation for $\{u,v\}$. For a graph $G$, we denote its vertex set by $V_G$ and its edge set by $E_G$. Two vertices $u$ and $v$ are called \emph{adjacent} if $uv \in E_G$. A \emph{path} in $G$ is a sequence of distinct vertices such that every two consecutive ones are adjacent. A cycle in $G$ is a path that connects two adjacent vertices (the edge between these vertices should not belong to the path, but it belongs to the cycle). A graph is \emph{acyclic} if it contains no cycles. A graph is \emph{connected} if any two vertices are joined by a path. A \emph{simple tree} is a simple connected acyclic graph. The path on $n$ vertices $\{1,2,\ldots,n\}$ is denoted by $P_n$.

The (open) \emph{neighborhood} $N_G(u)$ of a vertex $u\in V_G$ in a simple graph $G$ is the set of all vertices adjacent to $u$. 
The \emph{link neighborhood} of $u$ is the set of edges incident with $u$ and we denote it by $N'_G(u)$.  
The \emph{degree} of $u$ in $G$, denoted by $\deg_G u$, is the number of vertices adjacent to $u$ (or equivalently the number of edges incident with $u$), formally $\deg_G u=|N_G(u)|=|N'_G(u)|$. If the graph $G$ is clear from the context, we often omit the subscripts and write simply $N(u), N'(u)$ and $\deg u$.

In a \emph{cubic} graph every vertex is of degree 3, while in a \emph{subcubic} one, every vertex is of degree at most 3.

A \emph{covering projection} from a simple graph $G$ (guest) to a simple graph $H$ (host) is a mapping $f: V_G\to V_H$ such that for every vertex $u\in V_G$, the restriction $f|_{N_G(u)}$ is a bijection between $N_G(u)$ and $N_H(f(u))$. 

A \emph{partial covering projection} is defined analogously, but the restriction $f|_{N_G(u)}$ is required to be an injection between $N_G(u)$ and $N_H(f(u))$ for each $u\in V_G$. 

In this paper we consider a more general definition of undirected graphs. From now on, a vertex may be connected to itself by a loop (contributing by 2 to its degree), or it might be incident with a semi-edge (an edge-type object having only one end-vertex and contributing by 1 to the degree of this vertex). Moreover, edges, loops, and semi-edges may have associated \emph{multiplicities}. Next, we present the formal definition. Although we do not work with graphs containing loops in this paper, we provide the definition in this more general form to maintain consistency with other existing literature.

\begin{definition} 
A \emph{graph} is a triple $G=(V,\Lambda,\iota)$, where $V$ is the set of vertices, $\Lambda=E\cup L\cup S$  is the set of \emph{links} distinguished as \emph{edges}, \emph{loops} 
and \emph{semi-edges}, respectively, and $\iota: \Lambda \to \binom{V}2 \cup V$ is an \emph{incidence function}, such that $\iota(E)\subseteq \binom{V}2$ and $\iota(L\cup S)\subseteq V$.
\end{definition}

The degree of $u\in V$ is then $\deg u=|\{e\in E: u\in \iota(e)\}|+2|\{l\in L: u= \iota(l)\}|+|\{s\in S: u=\iota(s)\}|$. Each loop is counted twice, as we see it as a topological curve that has both endpoints in $u$. The \emph{link-neighborhood} corresponding to the links originating from $u$ is $N'(u)=\{e\in E: u\in \iota(e)\} \cup \{l\in L: u=\iota(l)\}\times[2] \cup \{s\in S: u=\iota(s)\}$.
This choice of link-neighborhood is conform with the above definition of the degree $\deg u=|N'(u)|$.
A vertex is called \emph{semi-simple} if it is incident with no loops, with no multiple edges and with at most one semi-edge. Most standard graph-theoretical notions extend naturally from simple graphs to general ones. The graph is called \emph{$d$-regular} if all of its vertices have degree $d$, and it is called \emph{regular} if it is $d$-regular for some $d$. Paths and cycles are considered as sequences of links, rather than sequences of vertices. A path is a sequence of edges and semi-edges such that any two consecutive ones share a vertex and all vertices appearing in the path are different; it follows that a path may start and/or end with a semi-edge, but all inner links are normal edges. Such a path is \emph{open} if it starts and ends with semi-edges, and it is \emph{half-open} if exactly one of the terminal links is a semi-edge. On top of simple cycles (i.e., cycles in the sense of cycles in simple graphs), any two parallel edges form a cycle of length $2$, and a loop is a cycle of length $1$. No cycle contains any semi-edge. A \emph{tree} is again defined as a connected acyclic graph. That means that trees have no multiple edges, and no loops, but they may contain semi-edges. A graph is called \emph{bipartite} if its vertex set can be partitioned into two disjoint sets, each of them being stable in the sense that it contains no links (not even semi-edges).    

\begin{definition}
A \emph{covering projection} $G\to H$ between two graphs is a mapping $f:(V_G\cup \Lambda_G) \to (V_H\cup \Lambda_H)$ such that
\begin{itemize}
\item vertices are mapped onto vertices, i.e., $f(V_G)\subseteq V_H$, and 
links are mapped onto links, i.e., $f(\Lambda_G)\subseteq \Lambda_H$,
%\item loops are mapped onto loops, i.e., %$f(L_G)\subseteq L_H$,
%\item semi-edges are mapped onto semi-edges, i.e., $f(S_G)\subseteq S_H$,
%\item incidences are preserved, i.e., for every $ \lambda\in \Lambda_G, f(\iota_G(\lambda))=\iota_H(f(\lambda))$,
\item for each edge $e\in E_H$ with $\iota(e)=\{u,v\}$, it holds that its preimage (also referred to as the lift) $f^{-1}(e)$ is a perfect matching between the sets $f^{-1}(u)$ and $f^{-1}(v)$
\item for each loop $l\in L_H$ with $\iota(l)=u$, it holds that its preimage (lift) $f^{-1}(l)$ is a disjoint union of cycles (including cycles of length 2 --- pairs of edges with the same end-points, and cycles of length 1 --- loops) spanning the set $f^{-1}(u)$,
\item for each semi-edge $s\in S_H$ with $ \iota(s)=u$, it holds that its preimage (lift) $f^{-1}(s)$ is a disjoint union of edges (i.e., a matching) and semi-edges spanning the set $f^{-1}(u)$. 
\end{itemize}
\end{definition}

We say that a graph $G$ covers a graph $H$ if it allows a covering projection from $G$ to $H$.

\begin{observation}
If $H$ is connected and $G$ is non-empty, then every covering projection $f:G\to H$ is surjective (both on vertices and links) and preserves incidences (i.e., $u\in \iota_G(e)$ if and only if $f(u)\in \iota_H(f(e))$ for every $u\in V_G$ and every $e\in \Lambda_G$. Moreover, semi-edges may only be mapped onto semi-edges, and loops may only be mapped onto loops.
\end{observation}

It also follows from the definition that a covering projection yields, for every vertex $u\in V_G$, a bijective mapping between $N'_G(u)$ and $N'_H(f(u))$ --- one only has to alternate on $l\times [2]$ when mapping a cycle to a loop $l$. Therefore, graph covering projections are also referred to as \emph{locally bijective homomorphisms}. A \emph{partial covering projection} loosens the bijection requirement to an injection. It is an incidence preserving mapping $f: G\to H$ such that  $N'_G(u)$ is mapped injectively into $N'_H(f(u))$ for every vertex $u$ of $G$.

When $f$ is a covering projection, then the lift $f^{-1}(x)$ is also referred to as the \emph{fiber} of (a vertex or a link) $x$.

A direct consequence of the well-known \emph{path lifting theorem}~\cite{k:Reidemeister32} is that when $G$ covers a connected graph $H$, then fibers of any two vertices in $H$ have the same cardinality. Such a covering is called \emph{$k$-fold}, where $k$ is the cardinality of the fibers.
The path lifting theorem also yields the following statement.

\begin{observation}\label{obs:partial-path-lift}
Let $f$ be a covering projection from $G$ to a connected graph $H$, let $S$ be a cutset of $G$ and let $A$ be a component of $G\setminus S$. Then for any two vertices $u,v\in V_H$, it holds that 
$|f^{-1}(u)\cap A|-|f^{-1}(v)\cap A|\le |S|$.
\end{observation}

\begin{proof}
As $H$ is connected, it contains a path $P$ connecting $u$ to $v$. As $f$ is locally bijective, $f^{-1}(P)$ is a disjoint union of paths in $G$~\cite{k:Reidemeister32}. 
At most $|S|$ of such paths may have one endpoint in $A$ and the other outside, hence the inequality follows.
\end{proof}

To simplify the hardness reductions, we often use the following proposition. For a graph $H$, the graph $H\times K_2$ is the graph with vertex set $\{u',u'':u\in V_H\}$ and edges formed as follows -- for every every normal edge with end-vertices $u$ and $v$, we add edges $u'v''$ and $u''v'$ to $H\times K_2$; for every semi-edge incident with vertex $u$, wee add the edge $u'u''$, and for every loop incident with $u$, we add two edges incident with $u'$ and $u''$.

\begin{proposition}[\!\cite{n:FK02}, Theorem 2.6]\label{prop:K2cover}
If $H$ is not bipartite, then a graph $G$ covers $H\times K_2$ if and only if $G$ is bipartite and covers $H$.
\end{proposition}

\begin{corollary}\label{cor:K2cover} The following hold for any graph $H$.
If {\sc $(H\times K_2)$-Cover} is NP-complete, then {\sc $H$-Cover} is NP-complete as well. And if {\sc $(H\times K_2)$-Cover} is NP-complete for simple input graphs, then so is {\sc $H$-Cover}.
\end{corollary}

The key building block of our NP-hardness reduction is the graph $G{:}w$ obtained from $G$ by splitting vertex $w$ into $k$ pendant vertices of degree 1. For each edge $e$ of $G$ incident with $w$, we formally keep this edge with the same name in $G{:}w$ and denote its pendant vertex of degree 1 by $w_e$. Then we have the following proposition.

\begin{proposition}[\!\cite{n:BFJKR24}, Proposition 2 (b--c)]
\label{prop:H_edge_colorable_properties} 
Let $H$ be a connected $k$-regular $k$-edge-colorable graph with no loops or semi-edges, let $G$ be a graph that covers $H$ and let $w$ be a vertex of $G$. Then
\begin{itemize}
\item in every partial covering projection from $G{:}w$ onto $H$, the pendant vertices $w_e, e\in N'_G(w)$ are mapped onto the same vertex of $H$;
\item in every partial covering projection from $G{:}w$ onto $H$, the pendant edges are mapped onto different edges (incident with the image of the pendant vertices).
\end{itemize}
\end{proposition}

This proposition describes cases of \cite[Proposition 2 (b--c)]{n:BFJKR24} applied to a broader class of graphs as only the assumed properties here are used in the proof in \cite{n:BFJKR24}.

In our reduction we will use the construction of a \emph{multicover}, introduced in \cite{n:KPT97}, and denoted by $M$, which can also be constructed for multigraphs~\cite{n:BFJKR24}.

\begin{proposition}[\!\cite{n:BFJKR24}, Proposition 1]
\label{prop:multicoverexistence}
Let $H$ be a connected $k$-regular $k$-edge-colorable graph with no loops or semi-edges. Then there exists a connected simple bipartite $k$-regular $k$-edge-colorable graph $M$ and a vertex $w \in V_M$ such that 
for every vertex $u\in V_H$ and for any bijection from $N'_M(w)$ onto $N'_H(u)$, there exists a covering projection from $M$ to $H$ which extends this bijection and maps $w$ to $u$.
\end{proposition}

Note that we further assume without loss of generality that $M$ is bipartite as otherwise $M\times K_2$ could be taken instead of it.

\begin{proposition}[\!\cite{n:BFJKR24}, Proposition 2 (a)]
\label{prop:multicoverproperties} 
Let $H$ be a connected $k$-regular $k$-edge-colorable graph with no loops or semi-edges. Then the graph $M{:}w$ constructed from the multicover $M$ of $H$ as above satisfies that
for every vertex $u\in V_H$ and every bijection $\sigma_u\colon N'_M(w) \to N'_H(u)$, there exists a partial covering projection of $M{:}w$ onto $H$ that extends $\sigma_u$ and maps each $w_e, e\in N'_M(w)$ to $x$.
\end{proposition}

Finally we make an observation about cycles of length 4 that we frequently use in our arguments. It directly follows from the fact that trees have no cycles.

\begin{observation}\label{obs:C4}
If $f: G\to T\na{k}$ is a covering projection and $G$ contains a $C_4$ as a subgraph, then 
at least one pair of opposite edges of such $C_4$ is mapped onto semi-edges.
\end{observation}

In other words, either all vertices of this $C_4$ are mapped to a vertex incident with two semi-edges or it is mapped onto two adjacent vertices, each adjacent with at least one semi-edge. 

\section{Paths}\label{sec:Paths}

We begin our study by examining paths. (Let us recall that $P_n$ stands for a path on $n$ vertices.) The $P_2\na{d}$-{\sc Cover} has been shown to be NP-complete in~\cite{n:BFHJK21-MFCS} for all $d\ge 3$. The $P_1\na{d}$-{\sc Cover} problem is equivalent to the $d$-edge-colorability of $d$-regular graphs, a problem well known to be NP-complete for every $d\ge 3$ for simple input graphs \cite{n:LevenGalil83}. Here we focus on covers of graphs $P_n\na{3}$ and $C_{2n}\na{3}$, depicted in Figure~\ref{fig:paths}.

\begin{figure}
\centering
\includegraphics[scale=1,page=1]{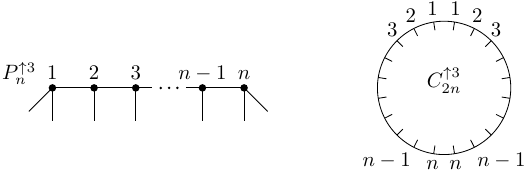}
\caption{Graphs $P_n\na{3}$ and $C_{2n}\na{3}$. The labels on $C_{2n}\na{3}$ indicate a covering projection to $P_n\na{3}$.}
\label{fig:paths}
\end{figure} 

Since semi-edges can be mapped only onto semi-edges (this holds for every covering), we get the following observations.

\begin{observation}\label{obs:fousata_cesta}
If $f:G \to P_n\na{3}$ is a covering projection and $G$ contains as a subgraph a path $Q$ where each vertex of $Q$ is incident with a semi-edge, then the sequence of images under $f$ along $Q$ coincides with a subsequence of the cyclic 
sequence of images given by the covering $C_{2n}\na{3}\to P_n\na{3}$ as depicted in Fig.~\ref{fig:paths}.
\end{observation}

\begin{figure}
\centering
\includegraphics[scale=1,page=2]{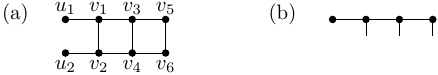}
\caption{Two auxiliary graphs.}
\label{fig:paths-aux}
\end{figure} 

\begin{observation}\label{obs:bezfousuaux}
If $f:G\to P\na{3}_n$ is a covering projection, $n\ge 3$, and $G$ contains the graph depicted in Fig.~\ref{fig:paths-aux} (a) as a subgraph, then $f(v_1)=f(v_2)$, $f(v_3)=f(v_4)$ and $f(v_5)=f(v_6)$. In other words,  
the three edges $v_1v_2$, $v_3v_4$ and $v_5v_6$ are mapped onto semi-edges.
Moreover, $f(u_1)=f(u_2)$.
\end{observation}

\begin{proof}
By Observation~\ref{obs:C4}, in each 4-cycle at least one pair of opposite edges is mapped onto semi-edges. This cannot be the pair containing the edge $v_1v_3$, as then $v_3v_5$ would be also mapped onto a semi-edge. This would consequently yield $n=2$, contrary to our assumptions.

Since $f(v_1)=f(v_2)$ and two neighbors of $v_1$, namely $v_2$ and $v_3$, have the same images as the two neighbors of $v_2$, namely $v_1$ and $v_4$, the remaining neighbors $u_1$ of $v_1$ and $u_2$ of $v_2$ are mapped on the same target to guarantee that $f$ is locally bijective between $N'(v_1)$ and $N'(f(v_1))$ and hence also between 
$N'(v_2)$ and $N'(f(v_2))=N'(f(v_1))$.
\end{proof}

\begin{lemma}\label{lem:bezfousu}
If $n\ge 3$ and $G$ is a graph of girth at least 5 such that each vertex of $G$ belongs to a subgraph isomorphic to the graph depicted in Fig.~\ref{fig:paths-aux}(b),
then $G$ covers $P\na{3}_n$ if and only if $G\times K_2$ covers $P\na{3}_n$. Furthermore, if all vertices of $G$ are semi-simple, $G\times K_2$ is simple and bipartite.
\end{lemma}

\begin{proof}
The forward implication follows from using the same target on both copies of every vertex of $G$.

For the reverse implication, by our assumptions, every vertex of $G\times K_2$ belongs to a subgraph isomorphic to the graph depicted in Fig.~\ref{fig:paths-aux} (a) as the product of $K_2$ with the graph depicted in Fig.~\ref{fig:paths-aux} (b). Finally, we apply Observation~\ref{obs:bezfousuaux} to conclude the proof.
\end{proof}

\begin{figure}
\centering
\includegraphics[page=3]{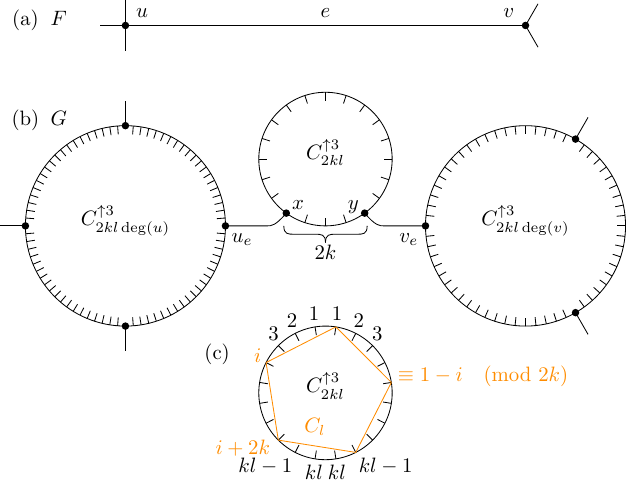}
\caption{Reduction from $C_l$-{\sc Hom} to $P\na{3}_{kl}$-{\sc Cover} for an odd $l$. In this particular case, $k=2$ and $l=5$.}
\label{fig:paths-homreduction}
\end{figure} 

\begin{lemma}\label{lem:oddpath}
The problem $P\na{3}_n$-{\sc Cover} is NP-complete for simple bipartite input graphs whenever $n$ is divisible by an odd $l\ge 3$. 
\end{lemma}

\begin{proof}
Let $n=kl$.
We reduce from $C_l$-{\sc Hom} which is NP-complete by Hell and Nešetřil's graph homomorphism dichotomy~\cite{b:HN90}. Let $F$ be a graph for which the existence of a homomorphism to $C_l$ is questioned. 

We construct a graph $G$ such that $G$ covers $P\na{3}_n$ if and only if $F$ allows a homomorphism to $C_l$. The elements of the construction are depicted in Fig.~\ref{fig:paths-homreduction} (b). For each vertex $u\in V_F$, we first insert a copy of $C\na{3}_{2kl\deg(u)}$. For each edge $e$ incident with $u$, we choose a distinct vertex of the cycle $C\na{3}_{2kl\deg(u)}$ and label it $u_e$, so that the chosen vertices are distributed around the cycle at distance $2kl$.

For each edge $e=uv\in E_F$, we insert a copy of $C\na{3}_{2kl}$, select two of its vertices $x$, and $y$ at distance $2k$ 
and merge the semi-edges stemming from $u_e$ and $x$ to an edge. We do this analogously for $v_e$ and $y$. This will be the edge gadget connecting the cycles representing $u$ and $v$.

If $f:G\to P\na{3}_{kl}$ is a covering projection, then by Observation~\ref{obs:fousata_cesta} we get that for every vertex $u\in V_F$, it holds that all vertices $u_e$ where $e\ni u$ are mapped by $f$ on the same vertex $i$ of $P\na{3}_{kl}$. Also 
their neighbors in the edge gadgets (denoted by $x$) are mapped onto $i$ as well. 
Note that when this pattern along $C\na{3}_{2kl\deg(u)}$ would be broken by mapping $f(u_ex)=\{i,i\pm 1\}$,
%\todo{JK: Nemelo by misto $\{1,2\}$ byt $\{i,i+1\}$, resp $\{i,i-1\}$? O vrcholu $u_e$ predpokladame, ze je mapovan na $i$. Mimochodem, tohle je misto, kde se mi zdalo potreba nechat hranu ve slozenych zavorkach.\par JF: Ano, opravil jsem a ješte dvojku o pár řádků níže. Koukněte prosím, jestli vám značení pomocí $\pm$ přijde srozumitelné.\par Všiml jsem si ještě  nekonzistentního značení hran $(u,v)$, v Obs 13, což jsem opravil, ale může to ještě být schované jinde. \par JK: Prosel jsem pomoci CtrlF a snad uz nikde jinde opravdu neni.}%
we fail to extend $f$ on the attached $C\na{3}_{2kl}$, because the neighbor of $x$ mapped onto $i\pm1$ cannot map its semi-edge as it was already used.
Such $i$ hence uniformly represents the image of $u$ under the desired homomorphism $F\to C_l$. 

Moreover, the vertices $x$ and $y$ in each edge gadget are at distance $2k$, so their images are labels of two vertices at distance $2k$ of $C\na{3}_{2kl}$. In particular, if $f(x)=i$, then:
$$f(y)\in
\begin{cases}
\{i+2k, i-2k\} & \text{ for } i\in[2k+1,kl-2k],\\
\{i+2k, 2k+1-i\} & \text{ for } i\in[1,2k],\\
\{2kl-2k+1-i, i-2k\} & \text{ for } i\in[kl-2k+1,kl].
\end{cases}
$$

Therefore, if any vertex $u_e$ satisfies $f(u_e)\equiv i\pmod{2k}$ with $i\in[1,2k]$, then for all other vertices $v_{e'}$, where $v\in V_E$, $e'\in E_F$ and $v\in e'$ holds that
$f(v_{e'})\in W_i$, where $W_i =\{j: j\equiv i \pmod{2k}\} \cup \{j\equiv 1-i \pmod{2k}\}$.

We construct the cycle $C_l$ on the vertex set $W_i$ where adjacent vertices have difference $2k$ or they form edges $\{i,2k+1-i\}$ and $\{kl-2k+i,kl+1-i\}$. See the orange cycle on Fig.~\ref{fig:paths-homreduction} (c) for an example.
Then the desired homomorphism $g\colon F\to C_l$ can be obtained by choosing $g(u)=f(u_e)$ for any choice $e\ni u$.

In opposite direction, given a homomorphism $g: F\to C_l$, we assume without loss of generality that 
$C_l$ is constructed on the set $W_i$ as discussed in the previous paragraph. Then we construct a covering 
$f\colon G\to  P\na{3}_{kl}$ first by letting $f(u_e)=g(u)$ for all $u\in V_F$ and $e\ni u$ and then extend this partial mapping to all the remaining vertices of vertex and edge gadgets.

Finally, when $k\ge 2$ we invoke Lemma~\ref{lem:bezfousu} to show that $G\times K_2$ is the desired simple bipartite graph which covers $P\na{3}_{kl}$ if and only if $F$ has a homomorphism to $C_l$.
To be able to apply Lemma~\ref{lem:bezfousu} also for the case $k=1$, we adjust the construction of 
$G$ so that in each edge gadget, we use $C\na{3}_{4l}$ instead of $C\na{3}_{2l}$. This is to ensure that along this cycle, $x$ and $y$ are connected by two paths of length $2l+2$ and $2l-2$. This does not alter the arguments of the hardness reduction, but shows that also vertices between $x$ and $y$ belong to subgraphs isomorphic to the graphs depicted in Fig.~\ref{fig:paths-aux} (b).
\end{proof}

\begin{figure}
\centering
\includegraphics[page=4]{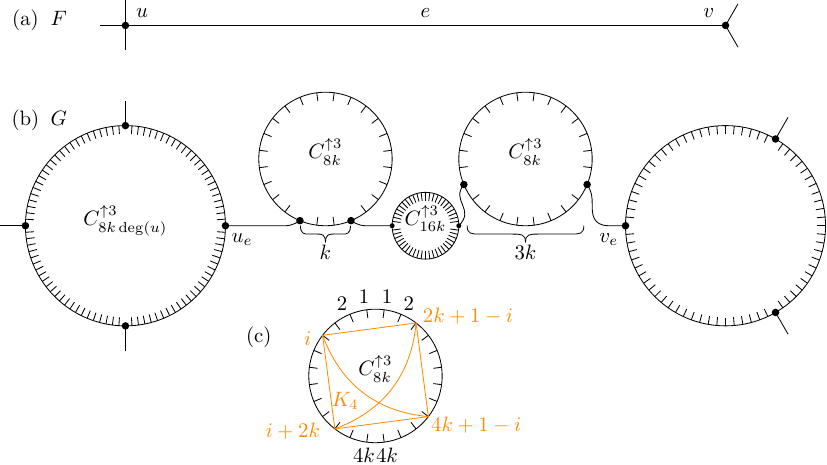}
\caption{Reduction from 4-{\sc Coloring} to $P\na{3}_{4k}$-{\sc Cover}, here $k=3$.}
\label{fig:paths-4colreduction}
\end{figure} 

\begin{lemma}\label{lem:4path}
The problem $P\na{3}_n$-{\sc Cover} is NP-complete for simple bipartite inputs whenever $n$ is divisible by 4. 
\end{lemma}

\begin{proof}[Proof idea]
We use the same design of the hardness reduction. Assume $n=4k$. The vertex gadgets are obtained by the identical construction as in the previous proof.

The edge gadget now consists of two copies of $C\na{3}_{8k}$ and one of $C\na{3}_{16k}$, see Fig.~\ref{fig:paths-4colreduction}. 
The copy of $C\na{3}_{16k}$ connected by antipodal vertices is involved only to 
guarantee that the resulting graph is bipartite.

The case analysis is similar and reused from~\cite[Theorem 6]{n:Fiala01}. The combination $\pm k\pm 3k$ allows us to label $u_e$ and $v_e$ by any combination of distinct labels from the set $\{i, i+2k, 2k+1-i, 4k+1-i\}$.
This set has always four elements as $i$ and $i+2k$ have opposite parity than $2k+1-i$ and $4k+1-i$.

Also, when $k=2$, then the first cycle $C\na{3}_{16}$ shall be prolonged to $C\na{3}_{32}$ to guarantee assumptions of  Lemma~\ref{lem:bezfousu}. (Observe that such adjustment is not necessary for $k=1$.)
\end{proof}

\begin{theorem}\label{thm:paths}
 For every $n\ge 1$, the problem $P\na{3}_n$-{\sc Cover} is NP-complete for simple input graphs. For $n\ge 2$, it is NP-complete even for simple bipartite input graphs.   
\end{theorem}

\begin{proof}
For $n=1$, the problem $P\na{3}_1$-{\sc Cover} is equivalent to 3-edge-colorability of cubic graphs, a well known NP-complete problem. Note, however, that this problem is polynomial-time solvable for bipartite graphs.

For $n=2$, the problem $P\na{3}_2$-{\sc Cover} is NP-complete for simple bipartite graphs, as shown in~\cite{n:BFHJK21-MFCS}.

If $n\ge 3$ is divisible by an odd number $l \ge 3$, the problem $P\na{3}_n$-{\sc Cover} is NP-complete for simple bipartite graphs by Lemma~\ref{lem:oddpath}. And finally, if $n\ge 4$ is a power of 2, $n$ is divisible by 4 and the problem $P\na{3}_n$-{\sc Cover} is NP-complete for simple bipartite graphs by Lemma~\ref{lem:4path}.
\end{proof}

\section{Trees of maximum degree 3}\label{sec:deg3}

In this section we describe an NP-hardness reduction from {\sc $3$-Edge-Coloring}. The reduction is heavily based on the reduction from~\cite{n:BFJKR24}. There a more general reduction from {\sc $k$-Edge-Coloring} is used to show that {\sc List-$H$-Cover} is NP-complete
whenever $H$ is a $k$-regular graph which contains at least one semi-simple vertex, and $k\ge 3$. (We recall that the input of {\sc List-$H$-Cover} is a graph together with lists of admissible targets for every vertex and link.) In fact, the reduction is such that the inputs constructed are simple graphs and the lists are of a very special format --- the lists of edges are general (i.e., equal to $\Lambda_H$) and the lists of vertices are either general (i.e., equal to $V_H$) or single-element (and these are all equal to the chosen semi-simple vertex of $H$). For the special case of 3-regular trees, the reduction goes in almost the same way, only a single vertex gadget will be altered to avoid the usage of lists. In the approach used in \cite{n:BFJKR22-IWOCA}, the lists were used to guarantee that certain vertices are always mapped onto a chosen semi-simple vertex. We show that under certain assumptions we do not need such a strong tool as lists, but we can guarantee that a certain vertex must map onto some semi-simple one. The claim is relaxed in the sense that we do not know which semi-simple vertex it will be, but nevertheless, the NP-completeness follows. 

We use an auxiliary bipartite graph $Q=Q(T,4)$ constructed as follows:

\begin{figure}
\centering
\includegraphics[scale=1,page=1]{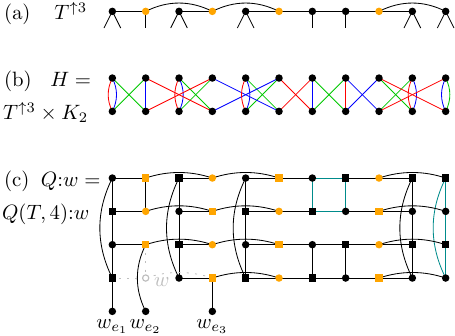}
\caption{(a) Example of a cubic tree $T\na{3}$, orange vertices are relevant; (b) The bipartite cubic multigraph $H$ covering $T\na{3}$ and its 3-edge coloring; (c) the graph $Q{:}w$, vertex shapes indicate classes of bi-partition, two types of cycles $C_4$ on the copies of irrelevant vertices are emphasized in cyan.}
\label{fig:equal3:q}
\end{figure} 

Take four copies of $T$ and denote the vertices in the $i$-th copy by $u_i$, for $u\in V_T$. For every $u\in V_T$ such that $\deg_T u=1$, form a $C_4$ by adding edges $u_iu_{(i \bmod 4)+1}$. For every $u\in V_T$ such that $\deg_T u=2$, add edges of a perfect matching of form $u_iu_{i+2}$, see Fig.~\ref{fig:equal3:q} (c).

We introduce the following notation: We call a vertex $u\in V_T$ \emph{relevant} (for the hardness reduction) if $\deg_T u=3$, or if $\deg_T u=2$ and it has a neighbor $v$ in $T$ such that $\deg_T v\ge 2$. The vertices which are not relevant are called {\em irrelevant}. Irrelevant vertices were already described in Observation~\ref{obs:C4} as possible images of vertices on a $C_4$.
Note that every relevant vertex is semi-simple, but the converse does not need to be true. In $Q$ as well as in $T \times K_2$, a vertex $u_i$ is called (ir)relevant if $u$ is (ir)relevant in $T$.

\begin{lemma}\label{lem:relevant}
The graph $Q$ covers $T\na{3}$ and this cover is $4$-fold. Moreover, in every covering projection $f:Q\rightarrow T\na{3}$, the image $f(u_i)$ of every relevant vertex $u_i$ is semi-simple.
\end{lemma}  

\begin{proof}
For the first statement, map every vertex $u_i, i\in[4]$ onto $u$, for every $u\in V_T$. The edges in $T$ are simple, and hence the image of every edge  $u_iv_i$ for $u\neq v$ is uniquely determined to be the edge $uv$. 
Then it suffices to show how the edges $u_iu_j$ are mapped. If $\deg_T u=2$, then in $T\na{3}$ the vertex $u$ is incident with a single semi-edge, hence the mapping of the two edges $u_1u_3$ and $u_2u_4$ is also uniquely determined. Analogously, when $u$ is a leaf in $T$, map the opposite edges on the $C_4$ on the four copies of $u$ to the same semi-edge incident with $u$. Then use the other semi-edge as the image of the other pair of opposite edges of such $C_4$.

We now focus on the second statement. Since $T\na{3}$ is connected, the preimages of the vertices of $T$ are of the same size $\frac{|V_Q|}{|V_{T\na{3}}|}=4$, and hence $f$ is a $4$-fold covering projection. Recall Observation~\ref{obs:C4} that the vertices of a 4-cycle in $Q$ can be mapped either all on the same vertex of $T\na{3}$ (in which case this vertex is a leaf with 2 semi-edges), or to 2 adjacent vertices such that each of them has at least one semi-edge pending on it. 

Therefore, every relevant vertex of $Q$ must be mapped onto a relevant vertex of $T\na{3}$. If $r$ is the number of relevant vertices of $T\na{3}$, the number of relevant vertices of $Q$ is $4r$, and hence $f$ must be mapping (ir)relevant vertices to the (ir)relevant ones. As noted above, relevant vertices are semi-simple.      
\end{proof}

\begin{figure}
\centering
\includegraphics[scale=1,page=3]{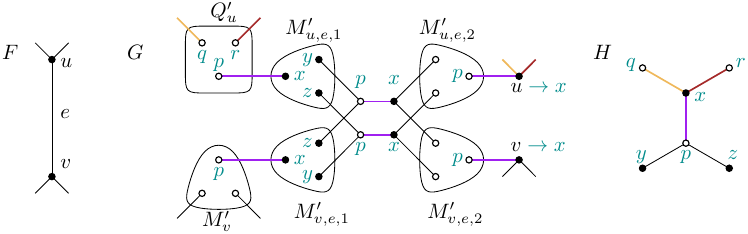}
\caption{Vertex and edge gadgets for the NP-hardness reduction from {\sc 3-Edge-Coloring}. Note that the graph $Q_w$ is in the resulting graph used only once. The cyan labels indicate the covering projection to $H$.}
\label{fig:equal3:red}
\end{figure} 

\begin{theorem}\label{thm:equal3}
For every simple tree $T$ of maximum degree $3$, the problem {\sc $T\na{3}$-Cover} is NP-complete, even for simple input graphs.
\end{theorem}

\begin{proof}
We reduce from {\sc 3-Edge-Coloring} and utilize Corollary~\ref{cor:K2cover} that for given cubic $F$ it suffices to construct a bipartite $G$ that covers $H=T\na{3}\times K_2$ if and only if $F$ is 3-edge colorable. See Fig.~\ref{fig:equal3:q} (b) for an example of $T\na{3}\times K_2$.

We first let $V_G=V_F$ to have for each vertex of $F$ its copy in $G$.

Then we choose an arbitrary vertex $u\in V_F$ and insert into $F$ a copy of the graph $Q{:}w$, denoted by $Q_u'$.
For every other vertex $v\in V_F$, we insert a copy $M_v'$ of the graph $M'=M{:}w$ obtained by splitting the vertex $w$ in a multicover $M$ of $T\na{3}$ as described in Propositions~\ref{prop:multicoverexistence} and \ref{prop:multicoverproperties}. 

For every edge $e\in E_F$, we insert into $G$ four copies of $M'$ and connect them to the vertex gadget and vertex copies as depicted in Fig.~\ref{fig:equal3:red}. This completes the construction of $G$. Graphs $Q$ and $M$ are bipartite, and thus $G$ is bipartite as well as indicated by the vertex colors in the figure.

Assume that $G$ covers $T\na{3}$. Since $G$ is bipartite, then by Proposition~\ref{prop:K2cover} it covers $H=T\na{3}\times K_2$. By Lemma~\ref{lem:relevant} the three edges stemming from $Q_u'$ are mapped onto three edges stemming from a relevant vertex $x\in V_H$, let us denote its neighbors $p,q,r$ which will represent three edge colors --- in Fig.~\ref{fig:equal3:red} we use purple, quince and ruby. 
Assume that the edge between $Q'_u$ and $M'_{u,e,1}$ is mapped onto 
$px$. Then the other two edges stemming from $M'_{u,e,1}$ are mapped onto $yp$ and $zp$, where $y$ and $z$ are the other two neighbors of $p$ in $H$.

Now from the perspective of $M'_{u,e,2}$, the three black vertices on edges stemming from $M'_{u,e,2}$ must be mapped onto the same neighbor of $p$. Since vertices $y,z$ are already used in $M'_{u,e,1}$, the only choice is $x$, see Fig.~\ref{fig:equal3:red}. As a consequence, the three edges between: $M'_{u,e,2}$ and $u$; $M'_v$ and $M'_{v,e,1}$; $M'_{v,e,2}$ and $v$ are all mapped onto $px$ as well.

In particular, both vertices $u$ and $v$ are mapped on $x$, and if we repeat the above argument for each edge gadget, we realize that the whole copy of $V_F$ in $G$ is mapped onto $x$. Moreover as the covering is locally bijective, the edges used as images on edges stemming from these vertices represent the desired edge colors. By our discussion, it follows that such a 3-edge coloring is well defined.

For the opposite direction, the edge coloring of $F$ can be directly transformed to a covering projection $G\to H$, 
using the relevant vertex $w\in V_T$ as the image of the copy of $V_F$, then extended on the elements of edge gadgets as described above, and finally extended inside the copy of $Q{:}w$ by the construction of $Q$ and inside each copy of $M'$ according to Proposition~\ref{prop:multicoverproperties}.

This completes the reduction of {\sc 3-Edge-Coloring} to {\sc $T\na{3}$-Cover}. Note that the size of the resulting graph $G$ is linear in the size of $F$ as we assume that $T$ is a constant parameter for the {\sc $T\na{3}$-Cover} problem.
\end{proof}

As we have noted above, the presence of $Q{:}w$ ensures that a relevant vertex with three distinct neighbors is enforced as the target vertex of all copies of the original vertices of $F$. In other words, instead of lists used in~\cite{n:BFJKR24}, we utilize the properties granted by Lemma~\ref{lem:relevant}.

\section{At least 4-regular trees}\label{sec:4-regular}

The polynomial reduction used in the proof of Theorem~\ref{thm:equal3} can be extended to $T\na{d}$ also for $d=\Delta\ge 4$, by reducing from {\sc$d$-Edge-Coloring} of $(d-1)$-uniform hypergraphs. On the other hand, for $d>\Delta$, there are no relevant vertices 
in $T\na{d}$ and hence a new reduction needs to be introduced. 

\begin{theorem}\label{thm:atleast4}
Let $T$ be a simple tree with at least 3 vertices and let $T\na{(k+1)}$ be a $(k+1)$-regular tree obtained from $T$ by adding semi-edges to its vertices (vertex $u$ gets $k+1-\deg_Tu$ semi-edges). If $k\ge 3$, then {\sc $T\na{(k+1)}$-Cover} is NP-complete even for simple input graphs. 
\end{theorem}

\begin{proof}
We reduce from {\sc $k$-Edge-Coloring} of $k$-regular graphs. For every fixed $k\ge 3$, it is an NP-complete problem~\cite{Leven-Galil1983}. 

A crucial gadget in the construction is the graph $Q=Q(T,2k)$, obtained as follows: Fix a leaf $a$ of $T$. Take $2k$ copies of $T$ and denote by $w_i$ the copy of a vertex $w$ in the $i$-th copy, for every $w\in V_T$. For every vertex $w\neq a$, add edges of a $(k+1-\deg_T w)$-regular bipartite graph with classes of bipartition $\{w_1,\dots,w_k\}$ and $\{w_{k+1},\dots,w_{2k}\}$. For the special vertex $a$, add edges of a $k$-regular $k$-edge-colorable graph on vertices $\{a_2,\ldots,a_{2k-1}\}$. Such $k$-regular $k$-edge-colorable graph exists, since $K_{2k-2}$ is Vizing class 1, i.e., its vertex set can be partitioned into $2k-3$ perfect matchings. We then take any $k$ of them, which is always possible since $2k-3\ge k$ for $k\ge 3$. 

\begin{figure}
\centering
\includegraphics[scale=1,page=1]{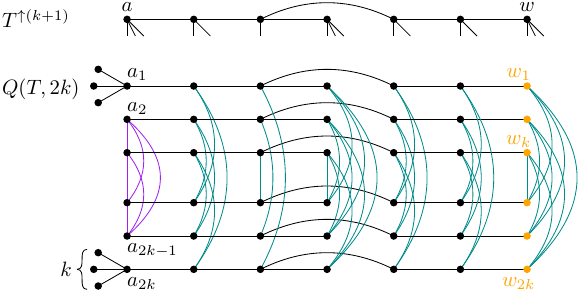}
\caption{Example of the construction of the graph $Q=Q(T,2k)$ from Theorem~\ref{thm:atleast4} for $k=3$. Horizontal layers correspond to copies of $T\na{(k+1)}$, cyan subgraphs are $(k+1-\deg_T w)$-regular bipartite and the purple is $k$-regular $k$-edge-colorable. Orange vertices form the fiber of a leaf $w$.}
\label{fig:atleast4}
\end{figure} 

Finally, add $k$ pendant edges to $a_1$, and $k$ pendant edges to $a_{2k}$, see Fig.~\ref{fig:atleast4}. Thus $Q$ has $2k|V_T|$ vertices of degree $k$ (we refer to them as the {\em core vertices} of $Q$) and $2k$ auxiliary vertices of degree 1. Recall that for every vertex $w\in V_T$, the subgraph of $Q$ induced by $\{w_1,w_2,\dots,w_{2k}\}$ is called the fiber of $w$ in $Q$. 

Suppose $f:Q\to T\na{(k+1)}$ is a partial covering projection. We make several observations.

\medskip\noindent
{\em Claim 1.} For every leaf $w\neq a$ of $T$, $f$ is constant on the $w$-fiber of $Q$, and the image of this fiber is a leaf of $T$.     

\medskip\noindent
{\em Proof of Claim 1.} The $w$-fiber is isomorphic to the complete bipartite graph $K_{k,k}$. Suppose an edge of the fiber is mapped onto an ordinary edge of $T$, say $f(w_1)=x$ and $f(w_{k+1})=y$, with $xy\in E_T$. Then none of $w_2,\ldots,w_k$ maps onto $x$, because $w_{k+1}$ has already one neighbor mapped onto $x$. Similarly, no $w_{k+2},\ldots,w_{2k}$ maps onto $y$. If any vertex of $\{w_2,\ldots,w_k\}$ were mapped onto another neighbor $x'$ of $y$ in $T$, $y$ would be the only common neighbor of $x$ and $x'$ in $T\na{(k+1)}$, and thus all  $\{w_{k+2},\ldots,w_{2k}\}$ have to be mapped onto $y$, which we already showed impossible. Hence for each $i\in[2,k]$ the edge $w_{k+1}w_i$ is mapped onto a semi-edge incident with $y$, and also $f(w_i)=y$. Similarly, we get $f(w_i)=x$ for $i\in[k+2,2k]$. But then $w_1$ (mapped to $x$) has $k-1\ge 2$ neighbors mapped onto $y$, a contradiction. Therefore, all edges $w_iw_j$ with $i\in[k], j\in[k+1,2k]$ are mapped onto semi-edges, in particular onto semi-edges incident with the same vertex $z$ of $T\na{(k+1)}$. Since $K_{k,k}$ is $k$-regular, this $z$ must be incident to (at least) $k$ semi-edges in $T\na{(k+1)}$, and hence be a leaf of $T$.     

Let us denote by $n_x$ the number of core vertices of $Q$ that are mapped to $x$ by $f$, for $x\in V_{T\na{(k+1)}}$. Since $\sum_{x\in V_{T\na{(k+1)}}}n_x = 2k|V_T|$, the average of $n_x$ is $2k$. By Observation~\ref{obs:partial-path-lift} we get: 

\medskip\noindent
{\em Claim 2.} For any $x,y\in V_{T\na{(k+1)}}$, $n_x$ and $n_y$ differ by at most 2. Moreover, $n_x\in\{2k-1,2k-,2k+1\}$ for every $x\in V_{T\na{(k+1)}}$. 

\medskip\noindent
{\em Claim 3.} The mapping $f$ is constant on every fiber. Moreover, the pendant edges of $Q$ are mapped onto semi-edges, and $f(a_1)=f(a_{2k})$ is a leaf of $T$.

\medskip\noindent
{\em Proof of Claim 3.} Suppose $T$ has $\ell$ leaves, $\ell\ge 2$. (Note also that since we assume that $T$ has at least 3 vertices, no two leaves of $T$ are adjacent.) In $Q$, there are $\ell-1$ leaf-fibers that induce complete bipartite graphs $K_{k,k}$. Each of them is mapped onto a leaf of $T$, and they must be mapped onto different ones, since otherwise some leaf $z$ would have $n_z\ge 4k$, contradicting Claim~2. So $\ell-1$ leaves are images of these fibers. Let $r$ be the last leaf of $T$. Root $T$ in $r$ and process the vertices of $T$ from the vertices farthest to the root. By induction we show that $f$ is constant on fibers and preserves the structure of $T$.

The base step of this induction are leaves of $T$, and for them the claim follows from Claim~1. 

In the inductive step, let us consider a vertex $x\in V_T$ with children $x^1,\dots,x^s$ with respective subtrees $T^1,\ldots,T^s$. By induction hypothesis, there are vertices $y^1,\dots,y^s\in V_T$ such that $f(y^i_j)=x^i$ for all $i\in[s]$ and $j\in[2k]$. The semi-edges incident with $x^i$'s are exactly covered by edges from the $y^i$-fiber, and the edges of $T^i$ are covered from the previously processed fibers. Hence each $y^i_j$ must have a neighbor that is mapped onto $x$ by $f$; this neighbor must come from the $j$-th copy of $T$, along a unique edge that has not been used so far. Thus each $y^i$ requests a fiber that is mapped onto $x$. If $s>1$, these fibers must coincide, since otherwise $n_x\ge 4k>2k+1$. Thus there is a unique vertex $y\in V_T$ such that $f(y_j)=x$ for all $j\in[2k]$, and $yy^i\in E_T$ for all $i\in[s]$. Since the edges within the $y$-fiber must cover the semi-edges incident with $x$, $\deg_T y=\deg_T x$.

When we reach the root $r$, the only fiber left available is the $a$-fiber. Hence $f(a_j)=r$ for all $j\in[2k]$. Also the ordinary edge incident with $r$ in $T$ is covered by the edges of the copies of $T$ (this comes from the induction), and hence the pendant edges incident with $a_1$ and $a_{2k}$ must map onto the semi-edges incident with $r$ in $T$.

\medskip
Now we are ready to complete the reduction. Given a connected $k$-regular simple graph $F$, take two copies $F_1, F_2$ of $F$ (with $u_1,u_2$ being the copies of a vertex $u\in V_F$ in $F_1$ and $F_2$, respectively) and $|V_F|$ copies $Q_{u}, u\in V_F$ of $Q$. Unify $a_1$ of  $Q_{u}$ with $u_1$, and $a_{2k}$ with $u_2$. Unify the pendant edges incident with $a_1$ in $Q_{u}$ with the edges of $F_1$ incident with $u_1$, and analogously for $a_{2k}$ and $F_2$. We claim that the graph $G$ constructed this way covers $T\na{(k+1)}$ if and only is $\chi'(F)=k$.

Suppose $f:G\to T\na{(k+1)}$ is a covering projection. Consider $u\in V_F$. Claim~3 states that there is a leaf $r$ of $T$ such that $f(u_1)=r$ and all edges of $F_1$ incident with $u_1$ are mapped onto semi-edges incident with $r$ in $T\na{(k+1)}$. If $v$ is a neighbor of $u$ in $F$, it follows that $f(v_1)=r$ as well. Moreover, since $F_1$ is connected, $f(V_{F_1})=\{r\}$ and all edges of $F_1$ are mapped onto the semi-edges incident with $r$. This mapping of edges of $F_1$ yields a proper $k$-edge-coloring of $F_1$. Hence $\chi'(F)=\chi'(F_1)=k$.

For the opposite implication, fix a proper $k$-edge-coloring of $F$, and use the names of the semi-edges incident with $a$ in $T\na{(k+1)}$ as the colors. Use the same edge-coloring on $F_1$ and $F_2$. Extend these partial covering projections to a covering projection $f:G\to T\na{(k+1)}$ in each $Q_{u}$ by setting $f(u_j)=u$ for every $j\in[2k]$ and $u\in V_T$, defining the mapping onto normal edges of $T\na{(k+1)}$ along the copies of $T$, and defining the mappings onto semi-edges as appropriate edge-colorings in the fibers (using the fact that bipartite graphs are of Vizing class 1, and the subgraph on the $a$-fiber was defined to be $k$-edge-colorable).     
\end{proof}

\section{Concluding remarks}

In this paper, we demonstrated that any regular tree with vertices of degree $d \ge 3$ yields an NP-complete graph covering problem even for simple input graphs. Since the 1997/2000 result on NP-hardness of covering simple regular graphs~\cite{n:KPT97,n:Fiala00b}, this is the first general NP-hardness theorem on graph covers encompassing a large class of regular target non-simple graphs with more than two vertices. For allowing non-simple target graphs in this theorem, we are paying by requesting them to be acyclic. 

As we prove the NP-hardness for simple input graphs, the result provides further confirmation of the Strong Dichotomy Conjecture~\cite{n:BFJKR24}, this time for the case of regular trees as target graphs. 

Besides the conjecture, whose proof remains the primary goal, we propose another open question that arises from our results: Suppose that $T$ is a tree of maximum degree $\Delta $. For $d\ge \Delta$, denote by $\overline{T\na{d}}$ a $d$-regular graph obtained from $T$ by adding semi-edges or loops so that every vertex of $\overline{T\na{d}}$ has degree $d$ (note that $\overline{T\na{d}}$ is not defined uniquely). %What is the computational complexity of graph covering problem for $\overline{T\na{d}}$?

\begin{conjecture}
If $\Delta\ge 3$ and $\overline{T\na{d}}$ is constructed as above, then the {\sc $\overline{T\na{d}}$-Cover} problem is NP-complete even for simple input graphs.
\end{conjecture}

\section*{Acknowledgment}
The authors thank the anonymous reviewers for valuable comments that helped increase the readability of the article.

\bibliographystyle{plainurl}
\bibliography{bib/barveni,bib/knizky,bib/nakryti,bib/sborniky,0-main}

\providecommand{\Dokovic}{{Ðoković}}
\begin{thebibliography}{10}

\bibitem{n:AFS91}
James Abello, Michael~R. Fellows, and John~C. Stillwell.
\newblock On the complexity and combinatorics of covering finite complexes.
\newblock {\em Australasian Journal of Combinatorics}, 4:103--112, 1991.
\newblock URL: \url{https://ajc.maths.uq.edu.au/pdf/4/ocr-ajc-v4-p103.pdf}.

\bibitem{n:Angluin80}
Dana Angluin.
\newblock Local and global properties in networks of processors.
\newblock {\em Proceedings of the 12th ACM Symposium on Theory of Computing},
  pages 82--93, 1980.
\newblock \href {https://doi.org/10.1145/800141.804655}
  {\path{doi:10.1145/800141.804655}}.

\bibitem{n:AG81}
Dana Angluin and A.~Gardiner.
\newblock Finite common coverings of pairs of regular graphs.
\newblock {\em Journal of Combinatorial Theory B}, 30:184--187, 1981.

\bibitem{n:Archdeacon02}
Dan Archdeacon.
\newblock Two graphs without planar covers.
\newblock {\em Journal of Graph Theory}, 41(4):318--326, 2002.

\bibitem{n:Biggs74}
Norman Biggs.
\newblock {\em Algebraic Graph Theory}.
\newblock Cambridge University Press, 1974.

\bibitem{n:Biggs81}
Norman Biggs.
\newblock Covering biplanes.
\newblock In {\em The theory and applications of graphs, Fourth International
  Conference, Kalamazoo}, pages 73--79. John Wiley \& Sons., 1981.

\bibitem{n:Biggs82}
Norman Biggs.
\newblock Constructing 5-arc transitive cubic graphs.
\newblock {\em Journal of London Mathematical Society II.}, 26:193--200, 1982.
\newblock \href {https://doi.org/10.1112/jlms/s2-26.2.193}
  {\path{doi:10.1112/jlms/s2-26.2.193}}.

\bibitem{n:Biggs84}
Norman Biggs.
\newblock Homological coverings of graphs.
\newblock {\em Journal of London Mathematical Society II.}, 30:1--14, 1984.
\newblock \href {https://doi.org/10.1112/jlms/s2-30.1.1}
  {\path{doi:10.1112/jlms/s2-30.1.1}}.

\bibitem{n:Bodlaender89}
Hans~L. Bodlaender.
\newblock The classification of coverings of processor networks.
\newblock {\em Journal of Parallel Distributed Computing}, 6:166--182, 1989.
\newblock \href {https://doi.org/10.1016/0743-7315(89)90048-8}
  {\path{doi:10.1016/0743-7315(89)90048-8}}.

\bibitem{n:BFJKS23-WG}
Jan Bok, Ji\v{r}\'{\i} Fiala, Nikola Jedli\v{c}kov\'{a}, Jan Kratochv\'{\i}l,
  and Michaela Seifrtov\'{a}.
\newblock Computational complexity of covering colored mixed multigraphs with
  degree partition equivalence classes of size at most two (extended abstract).
\newblock In {\em Graph-Theoretic Concepts in Computer Science: 49th
  International Workshop, WG 2023, Fribourg, Switzerland, June 28–30, 2023,
  Revised Selected Papers}, page 101–115, Berlin, Heidelberg, 2023.
  Springer-Verlag.
\newblock \href {https://doi.org/10.1007/978-3-031-43380-1_8}
  {\path{doi:10.1007/978-3-031-43380-1_8}}.

\bibitem{n:BFHJK21-MFCS}
Jan Bok, Jiří Fiala, Petr Hliněný, Nikola Jedličková, and Jan
  Kratochvíl.
\newblock Computational complexity of covering multigraphs with semi-edges:
  Small cases.
\newblock In Filippo Bonchi and Simon~J. Puglisi, editors, {\em 46th
  International Symposium on Mathematical Foundations of Computer Science},
  volume 202 of {\em LIPIcs}, pages 21:1--21:15. Schloss Dagstuhl ---
  Leibniz-Zentrum für Informatik, 2021.
\newblock \href {https://doi.org/10.4230/LIPIcs.MFCS.2021.21}
  {\path{doi:10.4230/LIPIcs.MFCS.2021.21}}.

\bibitem{n:BFJKR22-IWOCA}
Jan Bok, Jiří Fiala, Nikola Jedličková, Jan Kratochvíl, and Paweł
  Rzążewski.
\newblock List covering of regular multigraphs.
\newblock In Cristina Bazgan and Henning Fernau, editors, {\em Combinatorial
  Algorithms}, volume 13270 of {\em Lecture Notes in Computer Science}, pages
  228--242. Springer, 2022.
\newblock \href {https://doi.org/10.1007/978-3-031-06678-8_17}
  {\path{doi:10.1007/978-3-031-06678-8_17}}.

\bibitem{n:BFJKR24}
Jan Bok, Jiří Fiala, Nikola Jedličková, Jan Kratochvíl, and Paweł
  Rzążewski.
\newblock List covering of regular multigraphs with semi-edges.
\newblock {\em Algorithmica}, 86(3):782--807, 2024.
\newblock \href {https://doi.org/10.1007/S00453-023-01163-7}
  {\path{doi:10.1007/S00453-023-01163-7}}.

\bibitem{n:BFJKS21-FCT}
Jan Bok, Jiří Fiala, Nikola Jedličková, Jan Kratochvíl, and Michaela
  Seifrtová.
\newblock Computational complexity of covering disconnected multigraphs.
\newblock In Evripidis Bampis and Aris Pagourtzis, editors, {\em Fundamentals
  of Computation Theory}, volume 12867 of {\em Lecture Notes in Computer
  Science}, pages 85--99. Springer, 2021.
\newblock \href {https://doi.org/10.1007/978-3-030-86593-1_6}
  {\path{doi:10.1007/978-3-030-86593-1_6}}.

\bibitem{n:BilkaJKTV11}
Ondřej Bílka, Jozef Jirásek, Pavel Klavík, Martin Tancer, and Jan Volec.
\newblock On the complexity of planar covering of small graphs.
\newblock In Petr Kolman and Jan Kratochvíl, editors, {\em Graph-Theoretic
  Concepts in Computer Science}, volume 6986 of {\em Lecture Notes in Computer
  Science}, pages 83--94. Springer, 2011.
\newblock \href {https://doi.org/10.1007/978-3-642-25870-1_9}
  {\path{doi:10.1007/978-3-642-25870-1_9}}.

\bibitem{n:Chalopin05}
Jérémie Chalopin.
\newblock Local computations on closed unlabelled edges: the election problem
  and the naming problem.
\newblock In Peter Vojtáš, Mária Bieliková, Bernadette Charron-Bost, and
  Ondřej Sýkora, editors, {\em SOFSEM 2005: Theory and Practice of Computer
  Science}, volume 3381 of {\em Lecture Notes in Computer Science}, pages
  82--91. Springer, 2005.
\newblock \href {https://doi.org/10.1007/978-3-540-30577-4_11}
  {\path{doi:10.1007/978-3-540-30577-4_11}}.

\bibitem{n:ChMZ06}
Jérémie Chalopin, Yves Métivier, and Wiesław Zielonka.
\newblock Local computations in graphs: the case of cellular edge local
  computations.
\newblock {\em Fundamenta Informaticae}, 74(1):85--114, 2006.
\newblock \href {https://doi.org/10.5555/1231199.1231204}
  {\path{doi:10.5555/1231199.1231204}}.

\bibitem{n:ChalopinP11}
Jérémie Chalopin and Daniël Paulusma.
\newblock Graph labelings derived from models in distributed computing: {A}
  complete complexity classification.
\newblock {\em Networks}, 58(3):207--231, 2011.
\newblock \href {https://doi.org/10.1002/net.20432}
  {\path{doi:10.1002/net.20432}}.

\bibitem{packing_bipartite}
Jérémie Chalopin and Daniël Paulusma.
\newblock Packing bipartite graphs with covers of complete bipartite graphs.
\newblock {\em Discrete Applied Mathematics}, 168:40--50, 2014.
\newblock \href {https://doi.org/10.1016/j.dam.2012.08.026}
  {\path{doi:10.1016/j.dam.2012.08.026}}.

\bibitem{n:ChaplickFHPT15}
Steven Chaplick, Jiří Fiala, Pim van~'t Hof, Daniël Paulusma, and Marek
  Tesař.
\newblock Locally constrained homomorphisms on graphs of bounded treewidth and
  bounded degree.
\newblock {\em Theoretical Computer Science}, 590:86--95, 2015.

\bibitem{Corneil68}
Derek~G. Corneil.
\newblock {\em Graph Isomorphism}.
\newblock PhD thesis, University of Toronto, 1968.

\bibitem{n:CorneilG70}
Derek~G. Corneil and Calvin~C. Gotlieb.
\newblock An efficient algorithm for graph isomorphism.
\newblock {\em Journal of the Association for Computing Machinery}, 17:51--64,
  1970.
\newblock \href {https://doi.org/10.1145/321556.321562}
  {\path{doi:10.1145/321556.321562}}.

\bibitem{n:CM94}
Bruno Courcelle and Yves Métivier.
\newblock Coverings and minors: {A}pplications to local computations in graphs.
\newblock {\em European Journal of Combinatorics}, 15:127--138, 1994.
\newblock \href {https://doi.org/10.1006/eujc.1994.1015}
  {\path{doi:10.1006/eujc.1994.1015}}.

\bibitem{n:Djokovic74}
Dragomir~Ž. \Dokovic.
\newblock Automorphisms of graphs and coverings.
\newblock {\em Journal of Combinatorial Theory B}, 16:243--247, 1974.

\bibitem{n:Fiala00b}
Jiří Fiala.
\newblock {\em Locally injective homomorphisms}.
\newblock PhD thesis, Charles University, Prague, 2000.

\bibitem{n:Fiala01}
Jiří Fiala.
\newblock Cmputational complexity of covering cyclic graphs.
\newblock {\em Discrete Mathematics}, 235:87--94, 2001.

\bibitem{n:FHKT03}
Jiří Fiala, Pinar Heggernes, Petter Kristiansen, and Jan~Arne Telle.
\newblock Generalized {$H$}-coloring and {$H$}-covering of trees.
\newblock {\em Nordic Journal of Computing}, 10(3):206--224, 2003.

\bibitem{n:FialaKKN14}
Jiří Fiala, Pavel Klavík, Jan Kratochvíl, and Roman Nedela.
\newblock Algorithmic aspects of regular graph covers with applications to
  planar graphs.
\newblock In Javier Esparza, Pierre Fraigniaud, Thore Husfeldt, and Elias
  Koutsoupias, editors, {\em Automata, Languages and Programming}, volume 8572
  of {\em Lecture Notes in Computer Science}, pages 489--501. Springer, 2014.

\bibitem{arxiv1609.03013}
Jiří Fiala, Pavel Klavík, Jan Kratochvíl, and Roman Nedela.
\newblock Algorithmic aspects of regular graph covers, 2016.
\newblock URL: \url{https://arxiv.org/abs/1609.03013}, \href
  {https://arxiv.org/abs/1609.03013} {\path{arXiv:1609.03013}}.

\bibitem{n:FialaKKN18}
Jiří Fiala, Pavel Klavík, Jan Kratochvíl, and Roman Nedela.
\newblock {3-connected} reduction for regular graph covers.
\newblock {\em European Journal of Combinatorics}, 73:170--210, 2018.

\bibitem{n:FK02}
Jiří Fiala and Jan Kratochvíl.
\newblock Partial covers of graphs.
\newblock {\em Discussiones Mathematicae Graph Theory}, 22:89--99, 2002.

\bibitem{n:FK06}
Jiří Fiala and Jan Kratochvíl.
\newblock Locally injective graph homomorphism: {L}ists guarantee dichotomy.
\newblock In Fedor~V. Fomin, editor, {\em Graph-Theoretic Concepts in Computer
  Science}, volume 4271 of {\em Lecture Notes in Computer Science}, pages
  15--26. Springer, 2006.

\bibitem{n:FK08}
Jiří Fiala and Jan Kratochvíl.
\newblock Locally constrained graph homomorphisms --- structure, complexity,
  and applications.
\newblock {\em Computer Science Review}, 2(2):97--111, 2008.
\newblock \href {https://doi.org/10.1016/J.COSREV.2008.06.001}
  {\path{doi:10.1016/J.COSREV.2008.06.001}}.

\bibitem{n:FP05}
Jiří Fiala and Daniël Paulusma.
\newblock A complete complexity classification of the role assignment problem.
\newblock {\em Theoretical Computer Science}, 1(349):67--81, 2005.

\bibitem{n:Gardiner74}
Anthony Gardiner.
\newblock Antipodal covering graphs.
\newblock {\em Journal of Combinatorial Theory B}, 16:255--273, 1974.

\bibitem{getzler1998modular}
Ezra Getzler and Mikhail~M Kapranov.
\newblock Modular operads.
\newblock {\em Compositio Mathematica}, 110(1):65--125, 1998.
\newblock \href {https://doi.org/10.1023/A:1000245600345}
  {\path{doi:10.1023/A:1000245600345}}.

\bibitem{n:GriggsY92}
Jerrold~R. Griggs and Roger~K. Yeh.
\newblock Labelling graphs with a condition at distance 2.
\newblock {\em SIAM Journal on Discrete Mathematics}, 5(4):586--595, 1992.

\bibitem{b:HN90}
Pavol Hell and Jaroslav Ne\v{s}et\v{r}il.
\newblock On the complexity of {$H$}-colouring.
\newblock {\em Journal of Combinatorial Theory B}, 48:92--110, 1990.

\bibitem{n:Hlineny98}
Petr Hliněný.
\newblock {$K_{4,4}-e$} has no finite planar cover.
\newblock {\em Journal of Graph Theory}, 21(1):51--60, 1998.

\bibitem{n:HT04}
Petr Hliněný and Robin Thomas.
\newblock {On possible counterexamples to Negami's planar cover conjecture.}
\newblock {\em Journal of Graph Theory}, 46(3):183--206, 2004.

\bibitem{n:KPT94}
Jan Kratochvíl, Andrzej Proskurowski, and Jan~Arne Telle.
\newblock Complexity of graph covering problems.
\newblock In Ernst~W. Mayr, Gunther Schmidt, and Gottfried Tinhofer, editors,
  {\em Graph-Theoretic Concepts in Computer Science}, volume 903 of {\em
  Lecture Notes in Computer Science}, pages 93--105. Springer, 1994.
\newblock \href {https://doi.org/10.1007/3-540-59071-4_40}
  {\path{doi:10.1007/3-540-59071-4_40}}.

\bibitem{n:KPT97a}
Jan Kratochvíl, Andrzej Proskurowski, and Jan~Arne Telle.
\newblock Covering directed multigraphs {I.} {C}olored directed multigraphs.
\newblock In Rolf~H. Möhring, editor, {\em Graph-Theoretic Concepts in
  Computer Science}, volume 1335 of {\em Lecture Notes in Computer Science},
  pages 242--257. Springer, 1997.
\newblock \href {https://doi.org/10.1007/BFB0024502}
  {\path{doi:10.1007/BFB0024502}}.

\bibitem{n:KPT97}
Jan Kratochvíl, Andrzej Proskurowski, and Jan~Arne Telle.
\newblock Covering regular graphs.
\newblock {\em Journal of Combinatorial Theory, Series B}, 71(1):1--16, 1997.
\newblock \href {https://doi.org/10.1006/JCTB.1996.1743}
  {\path{doi:10.1006/JCTB.1996.1743}}.

\bibitem{n:KTT16}
Jan Kratochvíl, Jan~Arne Telle, and Marek Tesař.
\newblock Computational complexity of covering three-vertex multigraphs.
\newblock {\em Theoretical Computer Science}, 609:104--117, 2016.
\newblock \href {https://doi.org/10.1016/j.tcs.2015.09.013}
  {\path{doi:10.1016/j.tcs.2015.09.013}}.

\bibitem{n:KT00}
Petter Kristiansen and Jan~Arne Telle.
\newblock Generalized {$H$}-coloring of graphs.
\newblock In D.~T. Lee and Shang-Hua Teng, editors, {\em ISAAC}, volume 1969 of
  {\em Lecture Notes in Computer Science}, pages 456--466. Springer, 2000.

\bibitem{kwak2007graphs}
Jin~Ho Kwak and Roman Nedela.
\newblock Graphs and their coverings.
\newblock {\em Lecture Notes Series}, 17, 2007.

\bibitem{n:Leighton82}
Frank~Thomas Leighton.
\newblock Finite common coverings of graphs.
\newblock {\em Journal of Combinatorial Theory B}, 33:231--238, 1982.

\bibitem{n:LevenGalil83}
Daniel Leven and Zvi Galil.
\newblock {NP} completeness of finding the chromatic index of regular graphs.
\newblock {\em Journal of Algorithms}, 4:35--44, 1983.

\bibitem{Leven-Galil1983}
Daniel Leven and Zvi Galil.
\newblock {NP} completeness of finding the chromatic index of regular graphs.
\newblock {\em Journal of Algorithms}, 4(1):35--44, 1983.
\newblock \href {https://doi.org/10.1016/0196-6774(83)90032-9}
  {\path{doi:10.1016/0196-6774(83)90032-9}}.

\bibitem{n:LMZ93}
Igor Litovsky, Yves Métivier, and Wiesław Zielonka.
\newblock The power and the limitations of local computations on graphs.
\newblock In Ernst~W. Mayr, editor, {\em Graph-Theoretic Concepts in Computer
  Science}, volume 657 of {\em Lecture Notes in Computer Science}, pages
  333--345. Springer, 1992.
\newblock \href {https://doi.org/10.1007/3-540-56402-0_58}
  {\path{doi:10.1007/3-540-56402-0_58}}.

\bibitem{n:MacGillivrayS10}
Gary MacGillivray and Jacobus Swarts.
\newblock The complexity of locally injective homomorphisms.
\newblock {\em Discrete Mathematics}, 310(20):2685--2696, 2010.
\newblock Graph Theory --- Dedicated to Carsten Thomassen on his 60th Birthday.
\newblock \href {https://doi.org/10.1016/j.disc.2010.03.034}
  {\path{doi:10.1016/j.disc.2010.03.034}}.

\bibitem{n:MalnivcMP04}
Aleksander Malnič, Dragan Marušič, and Primož Potočnik.
\newblock Elementary abelian covers of graphs.
\newblock {\em Journal of Algebraic Combinatorics}, 20(1):71--97, 2004.
\newblock \href {https://doi.org/10.1023/B:JACO.0000047294.42633.25}
  {\path{doi:10.1023/B:JACO.0000047294.42633.25}}.

\bibitem{n:MalnicNS00}
Aleksander Malnič, Roman Nedela, and Martin Škoviera.
\newblock Lifting graph automorphisms by voltage assignments.
\newblock {\em European Journal of Combinatorics}, 21(7):927--947, 2000.
\newblock \href {https://doi.org/10.1006/eujc.2000.0390}
  {\path{doi:10.1006/eujc.2000.0390}}.

\bibitem{nedela_mednykh}
Alexander~D. Mednykh and Roman Nedela.
\newblock {\em Harmonic Morphisms of Graphs: Part I: Graph Coverings}.
\newblock Vydavatelstvo Univerzity Mateja Bela v Banskej Bystrici, 1st edition,
  2015.

\bibitem{n:NedelaS96}
Roman Nedela and Martin Škoviera.
\newblock Regular embeddings of canonical double coverings of graphs.
\newblock {\em Journal of Combinatorial Theory, Series B}, 67(2):249--277,
  1996.
\newblock \href {https://doi.org/10.1006/jctb.1996.0044}
  {\path{doi:10.1006/jctb.1996.0044}}.

\bibitem{n:Negami88}
Seiya Negami.
\newblock Graphs which have no planar covering.
\newblock {\em Bulletin of the Institute of Mathematics, Academia Sinica},
  16(4):377--384, 1988.

\bibitem{k:Reidemeister32}
Kurt Reidemeister.
\newblock {\em {Einf\"uhrung in die kombinatorische Topologie.}}
\newblock {Braunschweig: Friedr. Vieweg\&Sohn A.-G. XII, 209 S. }, 1932.

\bibitem{k:Ringel74}
Gerhard {Ringel}.
\newblock {\em {Map color theorem}}, volume 209.
\newblock Springer, Berlin, 1974.

\bibitem{arxiv1908.00830}
Sam Shepherd, Giles Gardam, and Daniel~J. Woodhouse.
\newblock Two generalisations of {L}eighton's theorem, 2019.
\newblock URL: \url{https://arxiv.org/abs/1908.00830}, \href
  {https://arxiv.org/abs/1908.00830} {\path{arXiv:1908.00830}}.

\bibitem{woodhouse_2021}
Daniel~J. Woodhouse.
\newblock Revisiting {L}eighton’s theorem with the {H}aar measure.
\newblock {\em Mathematical Proceedings of the Cambridge Philosophical
  Society}, 170(3):615--623, 2021.
\newblock \href {https://doi.org/10.1017/S0305004119000550}
  {\path{doi:10.1017/S0305004119000550}}.

\end{thebibliography}

\end{document}